\documentclass[envcountsame]{llncs}
\usepackage{graphicx}
\usepackage{amsmath}
\usepackage{amssymb}

\usepackage[nothing]{algorithm}
\usepackage[noend]{algorithmic}
\usepackage{mathptmx}
\usepackage{times}

\DeclareMathOperator*{\argmax}{argmax}

\newcommand{\noi}{\noindent}

\graphicspath{{../figures/}{.}}

\begin{document}
\pagestyle{plain}

\title{Optimal Interdiction of Unreactive Markovian Evaders}

\author{Alexander Gutfraind\inst{1} \and 
        Aric Hagberg\inst{2} \and
        Feng Pan\inst{3}}
        
\institute{
    Center for Applied Mathematics, 
    Cornell University, Ithaca, New York 14853
    \email{ag362@cornell.edu}
    \and 
    Theoretical Division,
    Los Alamos National Laboratory, Los Alamos, New Mexico 87545
    \email{hagberg@lanl.gov}
    \and 
    Risk Analysis and Decision Support Systems,
   Los Alamos National Laboratory, Los Alamos, New Mexico 87545
   \email{fpan@lanl.gov}
}

\maketitle

\setcounter{footnote}{0}

\begin{abstract}
The interdiction problem arises in a
variety of areas including military logistics, infectious disease
control, and counter-terrorism.  In the typical 
formulation of \textit{network} interdiction, the task of the interdictor is to
find a set of edges in a weighted network such that
the removal of those edges would maximally increase the
cost to an evader of traveling on a path through the network.

Our work is motivated by cases in which the evader has incomplete
information about the network or lacks planning time or computational
power, \textit{e.g.} when authorities set up roadblocks to catch
bank robbers, the criminals do not know all the roadblock locations
or the best path to use for their escape.  

We introduce a model of
network interdiction in which the motion of one or more evaders is
described by Markov processes and the evaders are assumed not
to react to interdiction decisions.  The interdiction 
objective is to find an edge set of size $B$, 
that maximizes the probability of capturing the evaders.

We prove that similar to the standard least-cost formulation for
deterministic motion 
this interdiction problem is also NP-hard.  But unlike that
problem our interdiction problem is submodular and 
the optimal solution can be approximated within $1-1/e$
using a greedy algorithm.  
Additionally, we exploit submodularity 
through a priority evaluation strategy 
that eliminates the linear complexity scaling in
the number of network edges and speeds up the solution
by orders of magnitude.
Taken together the results bring closer the goal of
finding realistic solutions to the interdiction problem on
global-scale networks.%
\footnote{{\em This article is released under Los Alamos National Laboratory LA-UR-09-00560}} \nonumber
\end{abstract}

\section{Introduction}
\label{sec:introduction}

Network interdiction problems have two opposing actors: an ``evader''
(\emph{e.g.} smuggler) and an ``interdictor'' (\emph{e.g.} 
border agent.)  The evader attempts to minimize some objective
function in the network, \emph{e.g.} the probability of capture while
traveling from network location $s$ to location $t$, while the
interdictor attempts to limit the evader's success by removing network nodes or
edges. Most often the interdictor has limited resources and can thus
only remove a very small fraction of the nodes or edges. The standard 
formulation
is the max-min problem where the interdictor plays first and chooses at most $B$ edges
to remove, while the evader finds the least-cost path on the remaining network.
This is known as the $B$ most vital arcs problem~\cite{Corley-1982-most}.

This least-cost-path formulation is not suitable for some
interesting interdiction scenarios.  Specifically in many practical problems there
is a fog of uncertainty about the underlying properties of the network
such as the cost to the evader in traversing an edge (arc, or link) in
terms of either resource consumption or detection probability.  In
addition there are mismatches in the cost and risk computations
between the interdictor and the evaders (as well as between different
evaders), and all agents have an interest in hiding their actions.  For
evaders, most least-cost-path interdiction models make optimal
assumptions about the evader's knowledge of the
interdictor's strategy, namely, the choice of interdiction set. In many
real-world situations evaders likely fall far short of the optimum.
This paper, therefore, considers the other limit case, 
which for many problems is more applicable, when the
evaders do not respond to interdictor's decisions.
This case is particularly useful for problems where the evader
is a process on the network rather than a rational agent.

Various formulations of the network interdiction problem have existed 
for many decades now.  The problem likely originated 
in the study of military supply chains and interdiction of transportation
networks~\cite{Mcmasters-1970-optimal,Ghare-1971-optimal}.  But in
general, the network interdiction problem applies to wide variety of
areas including control of infectious disease~\cite{Pourbohloul05}, and
disruption of terrorist networks~\cite{Farley03}. Recent interest in
the problem has been revived due to the threat of smuggling of nuclear
materials~\cite{Pan-2003-models}.  In this context interdiction of
edges might consist of the placement of special radiation-sensitive
detectors across transportation links.
For the most-studied formulation, that of max-min interdiction
described above~\cite{Corley-1982-most}, 
it is known that the problem is NP-hard~\cite{Ball89,Bar-noy-1995-complexity}
and hard to approximate~\cite{Boros06-inapproximability}.

\section{Unreactive Markovian Evader}
\label{sec:general-model}
The formulation of a stochastic model where the evader has limited or
no information about interdiction can be motivated by the following
interdiction situation.  Suppose bank robbers (evaders) want to escape
from the bank at node $s$ to their safe haven at node $t_1$ or node $t_2$. 
The authorities (interdictors) are able to position roadblocks at a few of
the roads on the network between $s$, $t_1$ and $t_2$.  The robbers might not
be aware of the interdiction efforts, or believe that they will be
able to move faster than the authorities can set up roadblocks.  They
certainly do not have the time or the computational resources to
identify the global minimum of the least-cost-path problem.

Similar examples are found in cases where
the interdictor is able to clandestinely remove edges or nodes
(\textit{e.g.} place hidden electronic detectors), or the evader
has bounded rationality or is constrained in strategic choices.
An evader may even have no intelligence of any kind and
represent a process such as Internet packet traffic that the
interdictor wants to monitor.  Therefore, our fundamental assumption
is that the evader does not respond to interdiction decisions.  This
transforms the interdiction problem from the problem of increasing the
evader's cost or distance of travel, as in the standard formulation,
into a problem of directly capturing the evader as explicitly defined
below. Additionally, the objective function acquires certain useful
computational properties discussed later.

\subsection{Evaders}

In examples discussed above, much of the challenge in interdiction
stems from the unpredictability of evader motion. Our approach is to
use a stochastic evader model to capture this 
unpredictability~\cite{Pan-2003-models,Gutfraind08markovian}.
We assume that an evader is traveling from a source node $s$ to a target
node $t$ on a graph $G(N,E)$ according to a guided random walk defined 
by the Markovian transition matrix ${\bf M}$; from node $i$
the evader travels on edge $(i,j)$ with probability $M_{ij}$. 
The transition 
probabilities can be derived, for example, from the cost and risk of
traversing an edge~\cite{Gutfraind08markovian}.  

Uncertainty in the evader's source
location $s$ is captured through a probability vector ${\bf a}$.
For the simplest case of an evader starting
known location $s$, $a_s=1$ and the rest of the $a_i$'s are $0$.
In general the probabilities can be distributed arbitrarily
to all of the nodes as long as $\sum_{i\in N} a_i=1$.
Given ${\bf a}$, the probability that the evader is at location $i$ 
after $n$ steps is the $i$'th entry in the vector 
${\boldsymbol\pi}^{(n)}={\bf a}{\bf M}^n$

When the target is reached the evader exits the network and therefore, 
$M_{tj} = 0$ for all outgoing edges from $t$ and also $M_{tt}=0$.
The matrix {\bf M} is assumed to satisfy the following condition:
for every node $i$ in the network either there is a positive probability of reaching
the target after a sufficiently large number of transitions, or 
the node is a dead end, namely $M_{ij} = 0$ for all $j$.
With these assumptions the Markov chain is absorbing and
the probability that the evader will eventually reach the target is $\leq 1$.
For equality to hold it is sufficient to have the extra conditions that the network 
is connected and that for all nodes $i\neq t$, $\sum_j{ M_{ij} }=1$
(see ~\cite{Grinstead97}.)

A more general formulation allows multiple evaders to traverse the network, 
where each evader represents a threat scenario or a particular adversarial group.
Each evader $k$ is realized with probability $w^{(k)}$ ($\sum_k w^{(k)}=1$) and is described by a
possibly distinct source distribution ${\bf a}^{(k)}$,
transition matrix ${\bf M}^{(k)}$, and target node $t^{(k)}$.  
This generalization makes it possible to represent any joint probability distribution $f(s,t)$ of
source-target pairs, where each evader is a slice of $f$ at a specific
value of $t$: ${\bf a}^{(k)}|_s={f(s,t^{(k)})}/\sum_s{f(s,t^{(k)})}$ and
$w^{(k)}=\sum_s{f(s,t^{(k)})}$.  In this high-level view, the evaders
collectively represent a stochastic process connecting pairs of
nodes on the network. This generalization has practical applications to problems of monitoring
traffic between any set of nodes when there is a limit on the number
of ``sensors''.  The underlying network could be \textit{e.g.}
a transportation system, the Internet, or water distribution 
pipelines.

\subsection{Interdictor}

The interdictor, similar to the typical
formulation, possesses complete knowledge about the network and evader
parameters ${\bf a}$ and ${\bf M}$.  Interdiction of an edge at index
$i,j$ is represented by setting $r_{ij}=1$ and $r_{ij}=0$ if the edge is
not interdicted. In general some edges are more suitable for
interdiction than others.  To represent this, we let $d_{ij}$ be the
interdiction efficiency, which is the probability that interdiction of
the edge would remove an evader who traverses it.

So far we have focused on the interdiction of edges,
but interdiction of nodes can be treated similarly as
a special case of edge interdiction in which all the edges leading to
an interdicted node are interdicted simultaneously.  For 
brevity, we will not discuss node interdiction further 
except in the proofs of Sec.~\ref{sec:proofs} where we consider
both cases.

\subsection{Objective function}

Interdiction of an unreactive evader is the problem of maximizing the
probability of stopping the evader before it reaches the target.
Note that the fundamental matrix for ${\bf M}$, using ${\bf I}$ to denote the identity matrix is
\begin{equation}
\label{eq:expected}
{\bf N} = {\bf I}+ {\bf M} + {\bf M}^2 + \cdots = ({\bf I}-{\bf M})^{-1} \,, 
\end{equation}
and {\bf N} gives all of the possible transition sequences between pairs of nodes before
the target is reached.  
Therefore given the starting probability ${\bf a}$, 
the expected number of times the evader reaches each node 
is (using (\ref{eq:expected}) and linearity of expectation)
\begin{equation}
    {\bf a}{\bf N}={\bf a} ({\bf I} - {\bf M})^{-1}\,.
    \label{eq:nodehits}
\end{equation}
If edge $(i,j)$ has been interdicted ($r_{ij}=1$) and the evader traverses
it then the evader will not reach $j$ with probability $d_{ij}$.
The probability of the evader reaching $j$ from $i$ becomes
\begin{equation}
\hat{M}_{ij} = M_{ij} - M_{ij} r_{ij} d_{ij}\,.
\end{equation}
This defines an interdicted version of the ${\bf M}$ matrix, the matrix ${\bf \hat{M}}$.  

The probability that a single evader does not reach the target
is found by considering the $t$'th entry in the vector $E$
after substituting ${\bf \hat{M}}$ for ${\bf M}$ in Eq.~(\ref{eq:nodehits}),
\begin{equation}
J({\bf a},{\bf M},{\bf r},{\bf d})=
1
-\left({\bf a}\left[{\bf I}-\left({\bf M}
-{\bf M}\odot {\bf r}\odot {\bf d}\right)\right]^{-1}
\right)_{t} 
\,,\label{eq:evader-cost}
\end{equation}
where the symbol $\odot$ means element-wise (Hadamard) multiplication. 
In the case of multiple evaders, the objective $J$ is a weighted sum,
\begin{equation}
J=\sum_{k}w^{(k)}J^{(k)} \,,\label{eq:weighted-cost}
\end{equation} where, 
for evader $k$,
\begin{equation}
J^{(k)}({\bf a}^{(k)},{\bf M}^{(k)},{\bf r},{\bf d})
=1-\left({\bf a}^{(k)}\left[{\bf I}-\left({\bf M}^{(k)}-{\bf M}^{(k)}
\odot {\bf r}\odot {\bf d}\right)\right]^{-1}\right)_{t^{(k)}} \,.
\label{eq:multiple-evader-cost}
\end{equation}

Equations (\ref{eq:evader-cost}) and (\ref{eq:weighted-cost})
define the \textit{interdiction probability}.  Hence the
\emph{Unreactive Markovian Evader} interdiction problem (UME) is 
\begin{equation}
\argmax_{{\bf r}\in F}\; J({\bf a},{\bf M},{\bf r},{\bf d}) \,,
\label{eq:UME}
\end{equation}
where $r_{ij}$ represents an interdicted edge
chosen from a set $F\subseteq 2^E$ of 
feasible interdiction strategies.
The simplest formulation is the case when interdicting an
edge has a unit cost with a fixed budget $B$ and 
$F$ are all subsets of the edge set $E$ of size at most $B$.
This problem can also be written as a mixed integer program as shown
in the Appendix.

Computation of the objective function can be achieved 
with $\sim\frac{2}{3}\left|N\right|^{3}$ operations 
for each evader, where $\left|N\right|$ is the number of nodes,
because it is dominated by the cost of Gaussian elimination solve
in Eq.~(\ref{eq:evader-cost}).
If the matrix ${\bf M}$ has special structure
then it could be reduced to $O(\left|N\right|^{2})$~\cite{Gutfraind08markovian}
or even faster.
We will use this evader model in the simulations, but in general
the methods of Secs.~\ref{sec:proofs} and~\ref{sec:blind} would work for 
any model that satisfies the hypotheses on ${\bf M}$ and even for non-Markovian evaders as long
as it is possible to compute the equivalent of the objective function in Eq.~(\ref{eq:evader-cost}).

Thus far interdiction was described as the removal of the evader from
the network, and the creation of a sub-stochastic process 
${\bf \hat{M}}$.  However, the mathematical formalism is 
open to several alternative interpretations.
For example interdiction could be viewed 
as redirection of the evader into a special absorbing state 
- a ``jail node''.
In this larger state space the evader even remains Markovian.
Since ${\bf \hat{M}}$ is just a mathematical device it is not even 
necessary for ``interdiction'' to change the physical traffic
on the network. In particular, in monitoring problems 
``interdiction'' corresponds to labeling of intercepted traffic
as ``inspected'' - a process that involves no removal or redirection.

\section{Complexity}
\label{sec:proofs}

This section proves technical results about the interdiction 
problem~(\ref{eq:UME})
including the equivalence in complexity of node and edge interdiction
and the NP-hardness of node interdiction (and therefore of edge interdiction).
Practical algorithms are found in the next section.

We first state the decision problem for~(\ref{eq:UME}).
\begin{definition}
{\bf UME-Decision}.

\noi \emph{Instance}: A graph $G(N,E)$, interdiction efficiencies $0\leq d_{i} \leq 1$ for each $i \in N$, budget $B \ge 0$,  
and real $\rho \geq 0$; a set $K$ of evaders, such that for each $k\in K$ there is a matrix ${\bf M}^{(k)}$ on $G$, 
a sources-target pair $({\bf a}^{(k)},t^{(k)})$ and a weight $w^{(k)}$.\\

\noi \emph{Question}:
Is there a set of (interdicted) nodes $Y$ of size $B$ such that 
\begin{equation} \label{eq:npc-ume}
    \sum_{k\in K}{w^{(k)}\left({\bf a}^{(k)}
      \left({\bf I}-\hat{{\bf M}}^{(k)}\right)^{-1}\right)_{t^{(k)}}} 
\le \rho ? 
\end{equation}
The matrix $\hat{{\bf M}}^{(k)}$ is constructed from ${\bf M}^{(k)}$ by replacing element $M^{(k)}_{ij}$ by
$M^{(k)}_{ij}(1-d_{i})$ for $i\in Y$ and each $(i,j)$ corresponding to edges $\in E$ leaving $i$.  
This sum is the weighted probability of the evaders reaching their targets.
\qed
\end{definition}

The decision problem is stated for node interdiction but the
complexity is the same for edge interdiction, as proved next.

\newtheorem{thm}{Theorem}
\newtheorem{lem}[thm]{Lemma}
\begin{lem}
Edge interdiction is polynomially equivalent to node interdiction.
\end{lem}
\begin{proof}
To reduce edge interdiction to node interdiction, take the 
graph $G(N,E)$ and construct $G'$ by splitting the edges.
On each edge $(i,j)\in E$ insert a node $v$ to
create the edges $(i,v),(v,j)$
and set the node interdiction efficiency 
$d_v=d_{ij}, d_i=d_j=0$, 
where $d_{ij}$ is the interdiction efficiency of $(i,j)$ in $E$. 

Conversely, to reduce node interdiction
to edge interdiction, construct from $G(N,E)$ a graph $G'$ by
representing
each node $v$ with interdiction efficiency $d_{v}$ by nodes
$i,j$, 
joining them with an edge $(i,j)$, and 
setting $d_{ij}=d_v$.
Next, change the transition matrix
${\bf M}$ of each evader such that all transitions into $v$ now move into
$i$ while all departures from $v$ now occur from $j$,
and $M_{ij}=1$. In particular, if $v$ was an evader's
target node in $G$, then $j$ is its target node in $G'$. \qed
\end{proof}

Consider now the complexity of node interdiction. One source of
hardness in the UME problem 
stems from the difficulty of avoiding the case where
multiple edges or nodes are interdicted on the same evader path - a
source of inefficiency.  This resembles the \emph{Set Cover} 
problem~\cite{Karp72}, where including an element in two sets is redundant in
a similar way, and this insight motivates the proof.

First we give the definition of the set cover decision problem.
\begin{definition} 
{\bf Set Cover.} For a collection $C$ of subsets of a finite set $X$,
and a positive integer $\beta$, does $C$ contain a cover of size $\leq
\beta$ for $X$? 
\qed
\end{definition}

Since {\em Set Cover} is NP-complete, 
the idea of the proof is to construct a network $G(N,E)$ 
where each subset $c\in C$ is represented
by a node of $G$, and each element $x_i\in X$ is represented by an evader. 
The evader $x_i$ is then made to traverse all nodes 
$\left\{c\in C | x_i\in c\right\}$.
The set cover problem is exactly problem of finding $B$ nodes 
that would interdict all of the evaders (see Fig.~\ref{fig:reduction}.)
\begin{figure}
\includegraphics[width=\columnwidth]{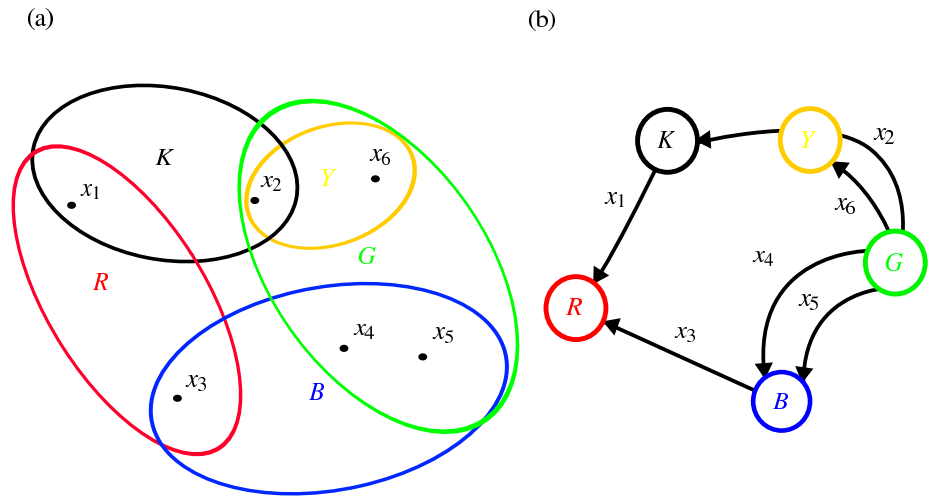}
\caption{\label{fig:reduction}
  Illustration of the reduction of Set Cover to UME-Decision. 
  (a) A set cover problem on elements
  $x_1\dots x_6\in X$ with subsets
  $K=\{x_1,x_2\},R=\{x_1,x_3\},B=\{x_3,x_4,x_5\},G=\{x_2,x_4,x_5,x_6\},Y=\{x_2,x_6\}$ contained in $X$.  
  (b) The induced interdiction problem with each subset represented by
  a node and each
  element by an evader. Each arrow indicates the path of a single
  evader.
}
\end{figure}

\begin{thm}
The UME problem is NP-hard even if $d_{i}=h$ (constant) $\forall$
nodes $i\in N$.
\end{thm}
\begin{proof}

First we note that for a given a subset $Y\subseteq N$ with $|Y| \le B$, 
we can update ${\bf M}^{(k)}$ and compute (\ref{eq:npc-ume}) to verify
\emph{UME-Decision} as a yes-instance. 
The number of steps is bounded by $O(|K||N|^3)$. Therefore,
\emph{UME-Decision} is in NP.

To show \emph{UME-Decision} is NP-complete, 
reduce \emph{Set Cover} with $X,C$ to \emph{UME-Decision} on a suitable graph $G(N,E)$. 
It is sufficient to consider just the special case where
all interdiction efficiencies are equal, $d_i=1$.  
For each $c\in C$, create a node $c$ in $N$. 
We consider three cases for elements $x\in X$; elements
that have no covering sets, elements that have one covering set,
and elements that have at least two covering sets.

Consider first all $x\in X$ which have at least two covering sets. For
each such $x$ create an evader as follows.  Let $O$ 
be any ordering of the collection of subsets covering
$x$.  Create in $E$ a Hamiltonian path of $|O|-1$ edges to
join sequentially all the elements of $O$, assigning the start, $a$
and end $t$ nodes in agreement with the ordering of
$O$.  Construct an evader transition matrix of size
$|C|\times|C|$ and give the evader transitions probability
$M_{ij}= 1$ iff $i,j \in C$ and $i<j$, and $=0$ otherwise.

For the case of zero
covering sets, that is, where $\exists x\in X$ such that $x\notin S$
for all $S\in C$, represent $x$ by an
evader whose source and target are identical: no edges are added to
$E$ and the transition matrix is ${\bf M}=0$.  Thus, $J$ in
Eq.~(\ref{eq:evader-cost}) is non-zero regardless of interdiction
strategy.  

For the case when $x$ has just one covering set, that is, when $\exists
x\in X$ such that there is a unique $c \in C$ with $x\in c$,
represent $c$ as two nodes $i$ and $j$ connected by an edge exactly as
in the case of more than one cover above.  After introducing $j$, add
it to the middle of the path of each evader $x$ if $i$ is in the path
of $x$, that is, if $c \in C$.  It is equivalent to supposing that $C$
contains another subset exactly like $c$.  This supposition does not
change the answer or the polynomial complexity of the given instance
of \emph{Set Cover}.  To complete the reduction, 
set $B = \beta$, $\rho = 0$, $X=K$, $w^{(k)}=1/|X|$
and $d_{i} = 1$, $\forall i \in N$.

Now assume \emph{Set Cover} is a yes-instance with a
cover $\hat{C}\subseteq C$.  We set the interdicted transition
matrix $\hat{M}^{(k)}_{ij} = 0$ for all $(i,j) \in E$
corresponding to  $c\in \hat{C}$, and all $k\in K$.
Since $\hat{C}$ is a cover
for $X$, all the created paths are disconnected, $\sum_{k\in
K}{({\bf a}^{(k)}({\bf I}-\hat{{\bf M}}^{(k)})^{-1})_{t^{(k)}}}
=0$ and \emph{UME-Decision} is an yes-instance.

Conversely, assume that \emph{UME-Decision} is a yes-instance. Let $Y$
be the set of interdicted nodes.  For $y\in Y$, there is element $y$
of $C$.  Since all the evaders are disconnected from their target and
each evader represents a element in $X$, $Y \subseteq C$ covers $X$
and $|Y| \le \beta$. Hence, \emph{Set Cover} is a yes-instance.
Therefore, \emph{UME-Decision} is NP-complete. \qed

\end{proof}

This proof relies on multiple evaders and it remains an open problem to show 
that UME is NP-hard with just a single evader.
We conjecture that the answer is positive because
the more general problem of interdicting a single unreactive evader
having an arbitrary (non-Markovian) path is NP-hard.
This could be proved by creating from a single such evader 
several Markovian evaders such that 
the evader has an equal probability of following
the path of each of the Markovian evaders in the proof above.

Thus far no consideration was given to the problem where
the cost $c_{ij}$ of interdicting an edge $(i,j)$ is not fixed but rather
is a function of the edge.  
This could be termed the ``budgeted'' case as opposed to the ``unit
cost'' case discussed so far.  However, the budgeted case is NP-hard 
as could be proved through reduction from the knapsack problem
to a star network with ``spokes'' corresponding to items.

\section{An Efficient Interdiction Algorithm}
\label{sec:blind}

The solution to the UME problem can be efficiently approximated using a greedy algorithm
by exploiting submodularity. In this section we prove that
the UME problem is submodular, construct
a greedy algorithm, and examine the algorithm's performance. 
We then show how to improve the algorithm's speed
by further exploiting the submodular structure using a ``priority''
evaluation scheme and ``fast initialization''.

\subsection{Submodularity of the interdiction problem}
In general, a function is called submodular if the rate of increase decreases
monotonically, which is akin to concavity.
\begin{definition}
A real-valued function on a space $S$, $f:S\to\mathbb{R}$ is \emph{submodular}~\cite[Prop. 2.1iii]{Nemhauser78} 
if for any subsets $S_{1}\subseteq S_{2}\subset S$ and any $x\in S\smallsetminus S_2$ it satisfies
\begin{equation}
f\left(S_{1}\cup \{x\}\right)-f\left(S_{1}\right)\geq f\left(S_{2}\cup \{x\}\right)-f\left(S_{2}\right) \,.\label{eq:submodular}
\end{equation}
\end{definition}

\begin{lem}
$J({\bf r})$ is submodular on the set of interdicted edges.
\end{lem}
\begin{proof}
First, note that it is sufficient to consider a single evader 
because in Eq.~(\ref{eq:weighted-cost}), $J({\bf r})$ is a convex combination
of $k$ evaders~\cite[Prop. 2.7]{Nemhauser78}.
For simplicity of notation, we drop the superscript $k$ in the rest of the proof.

Let $S = \{(i,j)\in E|r_{ij} = 1\}$ be the interdiction set
and let $J(S)$ be the probability of interdicting the evader using $S$,
and let $Q(p)$ be the probability of the evader taking a path $p$ to 
the target. 
On path $p$, the probability of interdicting the evader with 
an interdiction set $S$ is 
\begin{equation}
P(p|S) = Q(p)\left(1-\prod_{(i,j)\in p\cap S}{(1-d_{ij})}\right)\,.
\end{equation}
Moreover, 
\begin{equation}
    J(S)=\sum_{p}{P(p|S)}\,. \label{eq:evader-sum}
\end{equation}

If an edge $(u,v)\notin S$ is added to the interdiction set $S$ 
(assuming $(u,v)\in p$), the probability of interdicting 
the evader in path $p$ increases by 
$$
P(p|S\cup\{(u,v)\}) - P(p|S) = Q(p)d_{uv}\prod_{(i,j)\in p\cap S}{(1-d_{ij})}\,,
$$
which can be viewed as 
the probability of taking the path $p$ times the probability of being
interdicted at $(u,v)$ but not being interdicted elsewhere along $p$.
If $(u,v) \in S$ or $(u,v)\notin p$ then adding $(u,v)$ has, of
course, no effect: $P(p|S\cup\{(u,v)\}) - P(p|S) =0$. 

Consider now two interdiction sets $S_1$ and $S_2$ such that $S_1 \subset S_2$. 
In the case where $(u,v) \notin S_1$ and $(u,v)\in p$, we have 
\begin{eqnarray}
P(p|S_1\cup \{(u,v)\}) - P(p|S_1) 
& = &   Q(p)d_{uv}\prod_{(i,j)\in p\cap S_1}{(1-d_{ij})}\,,\label{eq:subm1}\\
& \ge & Q(p)d_{uv}\prod_{(i,j)\in p\cap S_2}{(1-d_{ij})}\,,\label{eq:subm2}\\
& \ge & P(p|S_2\cup \{(u,v)\})-P(p|S_2)\,.\label{eq:subm3}
\end{eqnarray}
In the above (\ref{eq:subm2}) holds because an edge $(u',v') \in
\left(S_2\smallsetminus S_1 \right)\cap p$ would contribute
a factor of $(1-d_{u'v'}) \le 1$.
The inequality (\ref{eq:subm3})
becomes an equality iff $(u,v) \notin S_2$. 
Overall
(\ref{eq:subm3})
holds true for any path and becomes an equality when $(u,v) \in S_1$. 
Applying the sum of Eq.~(\ref{eq:evader-sum}) gives
\begin{equation}
J(p|S_1\cup \{(u,v)\}) - J(p|S_1) \ge J(p|S_2\cup \{(u,v)\}) - J(p|S_2)\,,  
\end{equation}
and therefore $J(S)$ is submodular.\qed
\end{proof}

Note that the proof relies on the fact that the evader does
not react to interdiction. If the evader did react then it would no
longer be true in general that 
$P(p|S) = Q(p)\left(1-\prod_{(i,j)\in p\cap S}{(1-d_{ij})}\right)$
above.  Instead, the product may show explicit dependence on paths other than $p$, or
interdicted edges that are not on $p$. 
Also, when the evaders are not Markovian the proof is still valid because specifics of evader motion
are contained in the function $Q(p)$.

\subsection{Greedy algorithm}

Submodularity has a number of important theoretical and algorithmic
consequences. Suppose (as is likely in practice) that the edges are
interdicted incrementally such that the interdiction set
$S_{l}\supseteq S_{l-1}$ at every step $l$.
Moreover, suppose at each step, the interdiction set
$S_{l}$ is grown by adding the one edge that gives the greatest increase
in $J$. This defines a greedy algorithm, Alg.~\ref{al:greedy}.

\begin{algorithm}
    \caption{Greedy construction of the interdiction set $S$ with
      budget $B$ for a graph $G(N,E)$.
}
\begin{algorithmic}\label{al:greedy}

\STATE $S\leftarrow\varnothing$

\WHILE {$B>0$}

\STATE $x^*\leftarrow\varnothing$
\STATE $\delta^*\leftarrow -1$

\FORALL {$x\in E\smallsetminus S$}
\STATE $ \Delta(S,x):= J\left(S\cup\left\{ x\right\} \right)-J\left(S\right)$
\IF{$\Delta(S,x) > \delta^*$}
    \STATE $x^* \leftarrow \{x\}$
    \STATE $\delta^* \leftarrow \Delta(S,x)$
\ENDIF
\ENDFOR

\STATE $S\leftarrow S\cup x^*$

\STATE $B\leftarrow B-1$
\ENDWHILE
\STATE \textbf{Output}(S)
\end{algorithmic}
\end{algorithm}

The computational time is $O(B|N|^{3}|E|)$ for each evader,
which is strongly polynomial since $|B|\leq |E|$.
The linear growth in this bound as a function of the number of evaders
could sometimes be significantly reduced. 
Suppose one is interested in interdicting flow $f(s,t)$ 
that has a small number of sources 
but a larger number of targets.
In the current formulation the cost grows linearly in the number of
targets (evaders) but is independent of the number of sources.
Therefore for this $f(s,t)$ it is advantageous to reformulate UME 
by inverting the source-target relationship
by deriving a Markov process which describes 
how an evader moves from a given source $s$ to each of the targets.  
In this formulation the cost would be independent of the number
of targets and grow linearly in the number of sources.

\subsection{Solution quality}

The quality of the approximation can be bounded as a fraction of the
optimal solution by exploiting the submodularity property~\cite{Nemhauser78}.
In submodular set functions such as $J(S)$ there is an interference between
the elements of $S$ in the sense that sum of the individual contributions
is greater than the contribution when part of $S$.
Let $S_{B}^{*}$ be the optimal interdiction set with a budget $B$
and let $S_{B}^{g}$ be the solution with a greedy algorithm.
Consider just the first edge $x_1$ found by the greedy algorithm.
By the design of the greedy algorithm the gain from $x_1$ is greater
than the gain for all other edges $y$, including any of the edges in the optimal set $S^*$.
It follows that 
\begin{equation}
    \Delta(\varnothing,x_1) B \geq \sum_{y\in S_{B}^{*}} \Delta(\varnothing,y) \geq J(S_{B}^{*})\,.
\end{equation}
Thus $x_1$ provides a gain greater than the average gain for all the edges in $S_{B}^{*}$, 
\begin{equation}
\Delta(\varnothing,x_1)\geq\frac{J(S_{B}^{*})}{B}\,.
\end{equation}
A similar argument for the rest of the edges in $S_{B}^{g}$ gives the bound,
\begin{equation}
J(S_{B}^{g})\geq\left(1-\frac{1}{e}\right) J(S_{B}^{*}) \,,
\end{equation}
where $e$ is Euler's constant~\cite[p.268]{Nemhauser78}. Hence, the greedy algorithm achieves at least $63\%$ of the optimal
solution. 

This performance bound depends on the assumption that the cost of an
edge is a constant.  Fortunately, good discrete optimization
algorithms for submodular functions are known even for the case where
the cost of an element (here, an edge) is variable.  These algorithms
are generalizations of the simple greedy algorithm and provide a
constant-factor approximation to the
optimum~\cite{Khuller99,Krause05}.  Moreover, for any particular
instance of the problem one can bound the approximation ratio, and
such an ``online'' bound is often better than the ``offline'' \emph{a priori}
bound~\cite{Leskovec07}.

\subsection{Exploiting submodularity with Priority Evaluation}

In addition to its theoretical utility, 
submodularity can be exploited to compute the same solution much
faster using a priority evaluation scheme.
The basic greedy algorithm recomputes the objective function change
$\Delta(S_l,x)$ for
each edge $x\in E\smallsetminus S_l$ at each step $l$. 
Submodularity, however, implies that the gain $\Delta(S_l,x)$ from adding
any edge $x$ would be less than or equal to 
the gain $\Delta(S_k,x)$ computed at any earlier step $k<l$.
Therefore, if at step $l$ for some edge $x'$, 
we find that $\Delta(S_l,x')\ge\Delta(S_k,x$)
for all $x$ and any past step $k\leq l$, then $x'$ is
the optimal edge at step $l$; there is no need for further computation
(as was suggested in a different context~\cite{Leskovec07}.) 
In other words, one can use stale values of $\Delta(S_k,x)$ 
to prove that $x'$ is optimal at step $l$. 

As a result, it may not be necessary to compute $\Delta(S_l,x)$
for all edges $x\in E\smallsetminus S$ at every iteration. Rather,
the computation should prioritize the edges in descending order of
$\Delta(S_l,x)$. This ``lazy'' evaluation algorithm
is easily implemented with a priority queue which stores the gain
$\Delta(S_k,x)$ and $k$ for each edge where $k$ is the step at which
it was last calculated. (The step information $k$ determines whether
the value is stale.)

The priority algorithm (Alg.~\ref{al:priority}) combines lazy evaluation with the following fast initialization step.
Unlike in other submodular problems, in UME
one can compute $\Delta(\varnothing,x)$ simultaneously for all edges $x\in E$ because
in this initial step, $\Delta(\varnothing,x)$ is just the
probability of transition through edge $x$ multiplied by the interdiction efficiency $d_x$,
and the former could be found for all edges in just one operation.
For the ``non-retreating'' model of Ref. \cite{Gutfraind08markovian}
the probability of transition through $x=(i,j)$ is 
just the expected number of transitions though $x$ because 
in that model an evader moves through $x$ at most once.
This expectation is given by the $i,j$ element in ${\bf a}({\bf I}-{\bf M})^{-1} \odot {\bf M}$ 
(derived from Eq.~(\ref{eq:nodehits})).
The probability is multiplied by the weight of the evader and then by $d_x$:
$\Delta(\varnothing,x) = \sum_k{\left({\bf a}^{(k)}({\bf I}-{\bf M}^{(k)})^{-1}\right)_i M^{(k)}_{ij} w^{(k)} d_{x}}$.
In addition to these increments, for disconnected graphs the objective $J(S)$ also contains the constant 
term $\sum_k{w^{(k)}\left(\sum_{i\in Z^{(k)}}{a_i}\right)}$,
where $Z^{(k)}\subset N$ are nodes from which evader $k$ cannot reach his target $t^{(k)}$.

In subsequent steps this formula is no longer valid
because interdiction of $x$ may reduce the probability of motion through other
interdicted edges.
Fortunately, in many instances of the problem the initialization is the most
expensive step since it involves computing the cost for all edges in the graph.
As a result of the two speedups the number of cost evaluations could theoretically 
be linear in the budget and the number of evaders
and independent of the size of the solution space (the number of edges).

\begin{algorithm}[!ht]
    \caption{Priority greedy construction of the interdiction set $S$ with budget $B$}
\begin{algorithmic}\label{al:priority}

\STATE $S\leftarrow\varnothing$

\STATE $PQ\leftarrow\varnothing$ \COMMENT{Priority Queue: $(value,data,data)$}

\FORALL {$x = (i,j) \in E$}

\STATE $\Delta(x) \leftarrow $\COMMENT{The cost found using fast initialization}

\STATE $PUSH\left(PQ, \left(\Delta(x), x, 0\right)\right)$

\ENDFOR

\STATE $s\leftarrow 0$

\WHILE {$B>0$}

\STATE $s\leftarrow s+1$

\LOOP 

    \STATE $\left(\Delta(x), x, n\right) \leftarrow POP(PQ)$

    \IF{$n = s$}

        \STATE $S\leftarrow S\cup\{x\}$ 

        \STATE break

    \ELSE
 
        \STATE $\Delta(x) \leftarrow J\left(S\cup\left\{ x\right\} \right)-J\left(S\right)$
        
        \STATE $PUSH\left(PQ, \left(\Delta(x), x, s\right)\right)$

    \ENDIF

\ENDLOOP 

\STATE $B\leftarrow B-1$
\ENDWHILE
\STATE \textbf{Output}(S)
\end{algorithmic}
\end{algorithm}

The performance gain from priority evaluation can be very significant. In many
computational experiments, the second best edge from the previous step was the
best in the current step, and frequently only a small fraction of the
edges had to be recomputed at each iteration. 
In order to systematically gauge the improvement in performance, 
the algorithm was tested on $50$ synthetic interdiction problems. 
In each case, the underlying graph
was a $100$-node Geographical Threshold Graph (GTG),
a possible model of sensor or transportation
networks~\cite{Bradonjic-2007-wireless},
with approximately $1600$ directed edges (the threshold parameter was set at $\theta=30$). 
Most of the networks were connected.
We set the cost of traversing an edge to $1$,
the interdiction efficiency $d_{x}$ to $0.5$, $\forall x\in E$, and the budget to $10$. 
We used two evaders with uniformly distributed source nodes 
based on the model of \cite{Gutfraind08markovian} with an equal mixture of $\lambda=0.1$
and $\lambda=1000$. For this instance of the problem the priority algorithm
required an average of $29.9$ evaluations of the objective as compared to
$31885.2$ in the basic greedy algorithm - a factor of $1067.1$ speedup.

The two algorithms find the same solution, but the basic greedy
algorithm needs to recompute the gain
for all edges uninterdicted edges at every iteration, while the
priority algorithm can exploit fast initialization and stale computational values. 
Consequently, the former algorithm uses approximately $B|E|$ cost
computations, while the latter typically uses much fewer
(Fig.~\ref{fig:greedyPerf}a).

Simulations show that for the priority algorithm the number of edges did not
seem to affect the number of cost computations
(Fig.~\ref{fig:greedyPerf}b), in agreement with the theoretical limit.
Indeed, the only lower bound for the number of cost computations is
$B$ and this bound is tight (consider a graph with $B$ evaders each of
which has a distinct target separated from each evader's source by
exactly one edge of sufficiently small cost).  
The priority algorithm performance gains were
also observed in other example networks.%
\footnote{Specifically, the simulations were a two evader
problem on a grid-like networks consisting of a lattice (whose dimensions were
grown from $8$-by-$8$ to $16$-by-$16$) with random edges added at every node.
The number of edges in the networks grew from approximately $380$ to $1530$
but there was no increasing trend in the number of cost evaluations.}
\begin{figure}
\begin{center}
\includegraphics[width=0.5\columnwidth]{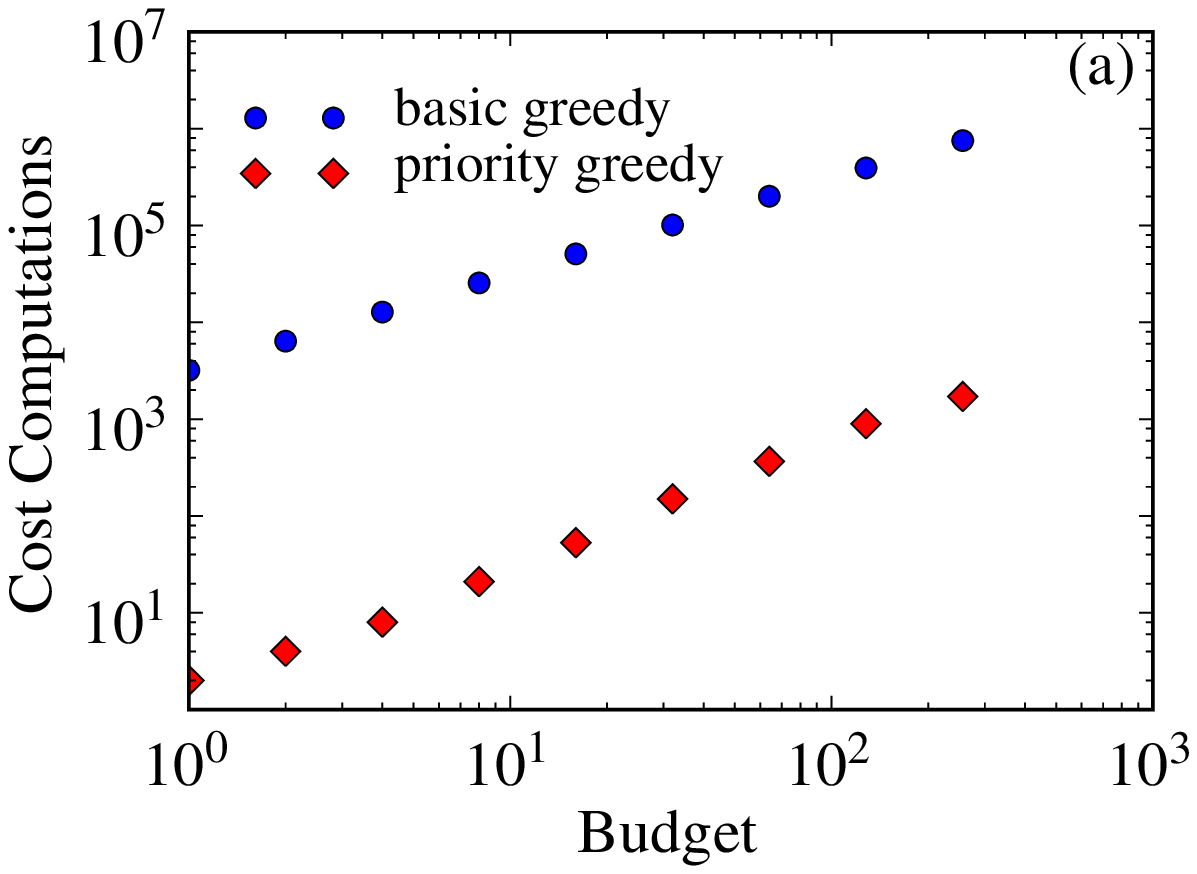}%
\includegraphics[width=0.5\columnwidth]{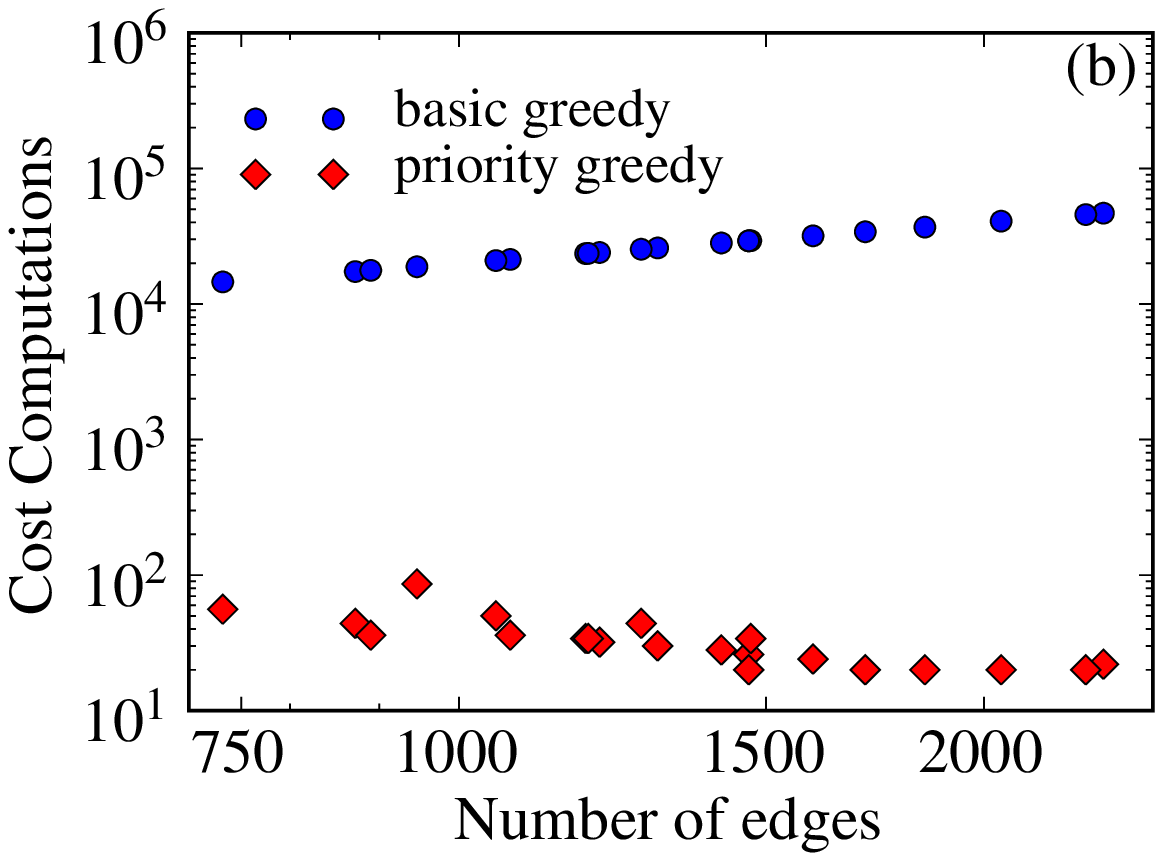}
\end{center}
\caption{Comparison between the basic greedy (blue circles) 
and the priority greedy algorithms (red diamonds) for
the number of cost evaluations as a function of 
(a) budget,
 and 
(b) number of edges.
In (a) each point is the average of $50$ network interdiction problems.
The average coefficient of variation 
(the ratio of the standard deviation to the mean)
is $0.10$ for basic greedy and $0.15$ for the priority greedy.
Notice the almost perfectly linear trends as a function of budget
(shown here on a log-log scale, the power $\approx 1.0$ in both.) 
In (b), the budget was fixed at $10$ and the number of edges was increased
by decreasing the connectivity threshold parameter
from $\theta=50$ to $\theta=20$ to 
represent, e.g., increasingly dense transportation networks.}
\label{fig:greedyPerf}
\end{figure}

The priority algorithm surpasses a
benchmark solution of the corresponding mixed integer program (See Appendix)
using a MIP solver running CPLEX (version 10.1) in consistency, 
time, and space.  For example, in runs on $100$-node GTG networks with
$4$ evaders and a budget of $10$, the priority algorithm terminates in $1$ to $20$ seconds, 
while CPLEX terminated in times ranging from under $1$ second to $9.75$ hours 
(the high variance in CPLEX run times, even on small problems, made
systematic comparison difficult.)
The difference in solution optimality was zero in the majority of runs.  
In the hardest problem we found (in terms
of its CPLEX computational time - $9.75$ hours), the priority
algorithm found a solution at $75\%$ of the optimum in less than $10$ seconds.

For our implementation, memory usage in the priority algorithm never
exceeded $300$MiB.
Further improvement could be made by re-implementing the priority 
algorithm so that it would
require only order $O(|E|)$ to store both 
the priority queue and the vectors of Eq.~(\ref{eq:evader-cost}).
In contrast, the implementation in CPLEX repeatedly used over $1$GiB for the
search tree.  As was suggested from the complexity proof, in runs where
the number of evaders was increased from $2$ to $4$ the computational
time for an exact solution grew rapidly.

\section{Outlook}

The submodularity property of the UME problem provides a rich source
for algorithmic improvement. In particular, there is room for more
efficient approximation schemes and practical value in their
invention.  Simultaneously, it would be interesting to classify the
UME problem into a known approximability class.
It would also be valuable to investigate various trade-offs
in the interdiction problem, such as the trade-off between quality and
quantity of interdiction devices.

As well, to our knowledge little is known about the accuracy of the assumptions
of the unreactive Markovian model or of the standard max-min model in various applications.
The detailed nature of any real instance of network interdiction 
would determine which of the two formulations is more appropriate.

\subsection*{Acknowledgments}
AG would like to thank Jon Kleinberg for inspiring lectures, David
Shmoys for a helpful discussion and assistance with software, and
Vadas Gintautas for support.  Part of this work was funded by the
Department of Energy at Los Alamos National Laboratory under contract
DE-AC52-06NA25396 through the Laboratory Directed Research and
Development Program.

\bibliographystyle{splncs}

\bibliography{interdict}

\begin{thebibliography}{10}

\bibitem{Corley-1982-most}
Corley, H.W., Sha, D.Y.:
\newblock Most vital links and nodes in weighted networks.
\newblock Oper. Res. Lett. \textbf{1}(4) (Sep 1982)  157 -- 160

\bibitem{Mcmasters-1970-optimal}
McMasters, A.W., Mustin, T.M.:
\newblock Optimal interdiction of a supply network.
\newblock Naval Research Logistics Quarterly \textbf{17}(3) (1970)  261--268

\bibitem{Ghare-1971-optimal}
Ghare, P.M., Montgomery, D.C., Turner, W.C.:
\newblock Optimal interdiction policy for a flow network.
\newblock Naval Research Logistics Quarterly \textbf{18}(1) (1971) ~37

\bibitem{Pourbohloul05}
Pourbohloul, B., Meyers, L., Skowronski, D., Krajden, M., Patrick, D., Brunham,
  R.:
\newblock Modeling control strategies of respiratory pathogens.
\newblock Emerg. Infect. Dis. \textbf{11}(8) (2005)  1246--56

\bibitem{Farley03}
Farley, J.D.:
\newblock Breaking {Al Qaeda} cells: A mathematical analysis of
  counterterrorism operations (a guide for risk assessment and decision
  making).
\newblock Studies in Conflict and Terrorism \textbf{26} (2003)  399--411

\bibitem{Pan-2003-models}
Pan, F., Charlton, W., Morton, D.P.:
\newblock Interdicting smuggled nuclear material.
\newblock In Woodruff, D., ed.: Network Interdiction and Stochastic Integer
  Programming.
\newblock Kluwer Academic Publishers, Boston (2003)  1--19

\bibitem{Ball89}
Ball, M.O., Golden, B.L., Vohra, R.V.:
\newblock Finding the most vital arcs in a network.
\newblock Oper. Res. Lett. \textbf{8}(2) (1989)  73--76

\bibitem{Bar-noy-1995-complexity}
Bar-Noy, A., Khuller, S., Schieber, B.:
\newblock The complexity of finding most vital arcs and nodes.
\newblock Technical report, University of Maryland, College Park, MD, USA
  (1995)

\bibitem{Boros06-inapproximability}
Boros, E., Borys, K., Gurevich, V.:
\newblock Inapproximability bounds for shortest-path network intediction
  problems.
\newblock Technical report, Rutgers University, Piscataway, NJ, USA (2006)

\bibitem{Gutfraind08markovian}
Gutfraind, A., Hagberg, A., Izraelevitz, D., Pan, F.:
\newblock Interdicting a {M}arkovian evader.
\newblock Preprint (2009)

\bibitem{Grinstead97}
Grinstead, C.M., Snell, J.L.:
\newblock Introduction to Probability. Second revised edn.
\newblock American Mathematical Society, USA (Jul 1997)

\bibitem{Karp72}
Karp, R.M.:
\newblock Reducibility among combinatorial problems.
\newblock In Miller, R.E., Thatcher, J.W., eds.: Complexity of Computer
  Computations.
\newblock New York: Plenum (1972)  85--103

\bibitem{Nemhauser78}
Nemhauser, G., Wolsey, L., Fisher, M.:
\newblock An analysis of the approximations for maximizing submodular set
  functions{-I}.
\newblock Mathematical Programming \textbf{14} (1978)  265--294

\bibitem{Khuller99}
Khuller, S., Moss, A., Naor, J.S.:
\newblock The budgeted maximum coverage problem.
\newblock Information Processing Letters \textbf{70}(1) (1999)  39--45

\bibitem{Krause05}
Krause, A., Guestrin, C.:
\newblock A note on the budgeted maximization on submodular functions.
\newblock Technical report, Carnegie Mellon University (2005) CMU-CALD-05-103.

\bibitem{Leskovec07}
Leskovec, J., Krause, A., Guestrin, C., Faloutsos, C., VanBriesen, J., Glance,
  N.:
\newblock Cost-effective outbreak detection in networks.
\newblock In: KDD '07: Proceedings of the 13th ACM SIGKDD international
  conference on Knowledge discovery and data mining, New York, NY, USA, ACM
  (2007)  420--429

\bibitem{Bradonjic-2007-wireless}
Bradonji\'c, M., Kong, J.S.:
\newblock Wireless ad hoc networks with tunable topology.
\newblock In: Forty-Fifth Annual Allerton Conference, UIUC, Illinois, USA
  (2007)  1170--1177

\end{thebibliography}

\section*{Appendix: Mixed integer program for UME}

In the unreactive Markovian evader interdiction (UME) problem
an evader $k\in K$ is sampled from a source distribution ${\bf a}^{(k)}$,
and moves to a sink $t^{(k)}$ with a path specified by the matrix ${\bf M}^{(k)}$.
This matrix is the Markov transition matrix with zeros in the row of the absorbing state (sink). 
The probability that the evader arrives at $t^{(k)}$ is $({\bf a}^{(k)}({\bf I}-{\bf M}^{(k)})^{-1})_{t^{(k)}}$
and is $1$ without any interdiction (removal of edges).

\newpage
\noi \textbf{Notation summary}
\begin{description}
{\setlength{\leftmargin}{3cm}}
\item $G(N,E)$: simple graph with node and edge sets $N$ and $E$, respectively.
\item $K$: the set of evaders.
\item $w^{(k)}$: probability that the evader $k$ occurs.
\item $a^{(k)}_i$: probability that node $i$ is the source node of evader $k$.
\item $t^{(k)}$: the sink of evader $k$.
\item ${\bf M}^{(k)}$: the modified transition matrix for the evader $k$.
\item $d_{ij}$: the conditional probability that interdiction of edge $(i,j)$ would remove an evader who traverses it.
\item $B$: the interdiction budget.
\item $\pi^{(k)}_i$: decision variable on conditional probability of node evader $k$ traversing node $i$.
\item $r_{ij}$: interdiction decision variable, $1$ if edge $(i,j)$ is interdicted and $0$ otherwise.
\end{description}

\begin{definition}\emph{Unreactive Markovian Evader} interdiction (UME) problem
\begin{eqnarray*}
\min_{{\bf r}} && H({\bf r}) = \sum_{k\in K}{w^{(k)} h^{(k)}({\bf r})}\,,\\
\mbox{\rm s.t. }  & & \sum_{(i,j)\in E}{r_{ij}} = B\,,\\
& & r_{ij} \in \{0,1\}, \hskip 5pt \forall (i,j) \in E,
\end{eqnarray*}
where
\begin{eqnarray}
 h^{(k)}({\bf r}) =& & \min_{\pi} \pi_{t^{(k)}}\,, \nonumber\\
\mbox{\rm s.t.}& & \pi^{(k)}_i -\sum_{(j,i)\in
  E}{(M^{(k)}_{ji}-M^{(k)}_{ji}d_{ji}r_{ji})\pi^{(k)}_j} = a^{(k)}_i,
\hskip 5pt \forall i \in N\,,
 \label{constraint_nonlinear} \\
 & & \pi^{(k)}_{i} \ge 0, \hskip 5pt \forall i \in N.
\end{eqnarray}
\end{definition}
The constraint (\ref{constraint_nonlinear}) is nonlinear. We can replace
this with a set of linear constraints, and the evader
problem becomes
\begin{subequations}\label{linear_form}
\begin{eqnarray}
 h^{(k)}({\bf r}) = && \min_{\pi,\theta} \pi_{t^{(k)}}\,, \nonumber\\
\mbox{\rm s.t.} & &\pi^{(k)}_i - \sum_{(j,i)\in E}{\theta^{(k)}_{ji}} = a^{(k)}_i, \hskip 5pt \forall i \in N\,,\nonumber\\
& &\theta^{(k)}_{ji} \ge M^{(k)}_{ji}\pi^{(k)}_j -M^{(k)}_{ji}d_{ji}r_{ji}, \hskip 5pt \forall (j,i) \in E\,, \label{constraint_upper}\\
& &\theta^{(k)}_{ji} \ge M^{(k)}_{ji}(1-d_{ji})\pi^{(k)}_j,   \hskip 5pt  \forall (j,i) \in E\,, \label{constraint_lower} \\
& & \theta^{(k)}_{ij} \ge 0, \hskip 5pt \forall (i,j) \in E\,, \nonumber\\
& & \pi^{(k)}_i \ge 0, \hskip 5pt \forall i \in N\,. \nonumber
\end{eqnarray}
\end{subequations}
If we set $r_{ij} = 0$, the constraint (\ref{constraint_upper}) is dominating (\ref{constraint_lower}), 
and $\theta_{ij}$ will take value $M^{(k)}_{ij}\pi^{(k)}_i$ at optimal because of the minimization. 
If we set $r_{ij} =1$, the constraint (\ref{constraint_lower}) is dominating since $\pi^{(k)}_j \le 1$.
Although formulation (\ref{linear_form}) has an additional variable
$\theta$, at the optimum the two formulations are equivalent because ${\bf
  \pi}$ and $\bf r$ have the same values.

\end{document}